\documentclass[11pt]{article}  

\usepackage{amssymb}
\usepackage{amsmath}
\usepackage{amsthm}
\usepackage{color}
\usepackage{colortbl}
\usepackage{graphicx}
\usepackage{hyperref}
\usepackage{fullpage}
\usepackage{enumerate}
\usepackage{xspace}
\usepackage{algorithm}
\usepackage{caption}
\usepackage{subcaption}
\usepackage{algpseudocode}
\usepackage{paralist}

\usepackage{lmodern}
\usepackage[T1]{fontenc}

\newcommand{\R}{\mathbb{R}}

\newcommand{\E}{\mathbb{E}}
\newcommand{\Ex}{\mathbb{E}}

\newcommand{\OPT}{\mathrm{OPT}}

\newtheorem{theorem}{Theorem}
\newtheorem{lemma}[theorem]{Lemma}
\newtheorem{corollary}[theorem]{Corollary}
\newtheorem{claim}[theorem]{Claim}

\newtheorem{observation}[theorem]{Observation}

\newcommand{\etal}{\emph{et al.}\xspace}
\newcommand{\eps}{\epsilon}
\DeclareMathOperator{\argmax}{argmax}

\begin{document}
\title{A Nearly-linear Time Algorithm for Submodular Maximization with a Knapsack Constraint}

\author{
Alina Ene\thanks{Department of Computer Science, Boston University, {\tt aene@bu.edu}.}
\and
Huy L. Nguy\~{\^{e}}n\thanks{College of Computer and Information Science, Northeastern University, {\tt hlnguyen@cs.princeton.edu}.} 
}
\date{}

\maketitle

\begin{abstract}
We consider the problem of maximizing a monotone submodular function subject to a knapsack constraint. Our main contribution is an algorithm that achieves a nearly-optimal, $1 - 1/e - \eps$ approximation, using $(1/\eps)^{O(1/\eps^4)} n \log^2{n}$ function evaluations and arithmetic operations. Our algorithm is impractical but theoretically interesting, since it overcomes a fundamental running time bottleneck of the multilinear extension relaxation framework. This is the main approach for obtaining nearly-optimal approximation guarantees for important classes of constraints but it leads to $\Omega(n^2)$ running times, since evaluating the multilinear extension is expensive. Our algorithm maintains a fractional solution with only a constant number of entries that are strictly fractional, which allows us to overcome this obstacle.  
\end{abstract}

\section{Introduction}
\label{sec:intro}

A set function $f:2^{V}\rightarrow \R$ is {\em submodular} if for every $A, B \subseteq V$, we have $f(A)+f(B) \ge f(A\cup B) + f(A\cap B)$. Submodular functions naturally arise in a variety of contexts, both in theory and practice. Submodular functions capture many well-studied combinatorial functions including cut functions of graphs and digraphs, weighted coverage functions, as well as continuous functions including the Shannon entropy and log-determinants. Submodular functions are used in a wide range of application domains from machine learning to economics. In machine learning, it is used for document summarization~\cite{LinB10}, sensor placement~\cite{KrauseSG08}, exemplar clustering~\cite{GomesK10}, potential functions for image segmentation~\cite{JegelkaB11}, etc. In an economics context, it can be used to model market expansion~\cite{DughmiRS12}, influence in social networks~\cite{Kempe2003}, etc. The core mathematical problem underpinning many of these applications is the meta problem of maximizing a submodular objective function subject to some constraints. 

A common approach to submodular maximization is a two-step framework based on the multilinear extension $F$ of $f$, a continuous function that extends $f$ to the domain $[0,1]^V$. The program first (1) maximizes $F(x)$ subject to a continuous relaxation of the constraint and then (2) rounds the solution $x$ to an integral vector satisfying the constraint. This paradigm has been very successful and it has led to the current best approximation algorithms for a wide variety of constraints including cardinality constraints, knapsack constraints, matroid constraints, etc. One downside with this approach is that in general, evaluating the multilinear extension is expensive and it is usually approximately evaluated. To achieve the desirable approximation guarantees, the evaluation error needs to be very small and in a lot of cases, the error needs to be $O(n^{-1})$ times the function value. Thus, even an efficient algorithm with $O(n)$ queries to the multilinear extension would require $\Omega(n^2)$ running time.

In this work, we develop a new algorithm that achieves $1-1/e-\eps$ approximation for maximizing a monotone submodular function subject to a knapsack constraint. The basic approach is still based on the multilinear extension but the algorithm ensures that the number of fractional coordinates is constant, which allows evaluating the multilinear extension exactly in constant number of queries to the original function. This approach allows us to bypass the obstructions discussed above and get nearly linear running time.

\begin{theorem}
  There is an algorithm for maximizing a monotone submodular function subject to a knapsack constraint that achieves a $1 - 1/e - \eps$ approximation using $(1/\eps)^{O(1/\eps^4)} n \log{n}$ function evaluations and $(1/\eps)^{O(1/\eps^4)} n \log^2{n}$ arithmetic operations.
\end{theorem}

For simplicity, when stating running times, we assume that each call to the value oracle of $f$ takes constant time, since for the algorithms discussed the number of evaluations dominates the running time up to logarithmic factors.  Previously, Wolsey~\cite{wolsey1982maximising} gives an algorithm with a $1 - 1/e^{\beta} \approx 0.35$, where $\beta$ is the unique root of the equation $e^x = 2 - x$. Building on the work of Khuller \etal for the maximum $k$-coverage problem~\cite{khuller1999budgeted}, Sviridenko~\cite{Sviridenko04} gives an algorithm with a $1 - 1/e$ approximation that runs in $O(n^5)$ time. Badanidiyuru and Vondrak \cite{Badanidiyuru2014} give an algorithm with a $1 - 1/e - \eps$ approximation running in $n^2 ({\log{n} / \eps})^{O(1/\eps^8)}$ time. Our work builds on \cite{Badanidiyuru2014} and we discuss the relationship between the two algorithms in more detail in Section~\ref{sec:techniques-knapsack}.

Kulik \etal \cite{kulik2013approximations} obtain a $1-1/e-\eps$ approximation for $d$ knapsack constraints in time $\Omega(n^{d/\eps^4})$ that comes from enumerating over $d/\eps^4$ items. The techniques in this paper could likely be extended to obtain an algorithm for the continuous problem of maximizing the multilinear extension subject to $d$ knapsack constraints, with a running time that is exponential in $d$ and nearly-linear in $n$. We leave it as an open problem whether the rounding can also be extended to multiple knapsack constraints.

\smallskip
{\bf Remark on the algorithm of \cite{Badanidiyuru2014}.} We note that there are some technical issues in the algorithm proposed in \cite{Badanidiyuru2014} for a knapsack constraint. The main issue, which was pointed out by Yoshida~\cite{Yoshida16}, arises in the partitioning of the items into large and small items: an item $e$ is small if it has value $f(\{e\}) \leq \eps^6 f(\OPT)$ and cost $c_e \leq \eps^4$, and it is large otherwise. The algorithm enumerates the marginal values of the large items and thus the set of large items was intended to be of size $\mathrm{poly}(1/\eps)$. But this may not be true in general, as there could be many items in $\OPT$ with singleton value greater than $\eps^6 f(\OPT)$. On the other hand, the assumption that the small items have small singleton values is crucial to ensuring that the algorithm obtains a good value from the small items. Another issue arises in the rounding algorithm. The fractional solution is rounded using a rounding algorithm for a partition matroid that treats the parts independently. But in this setting an item participates in several parts and we need to ensure that it is not selected more than once.

\subsection{Our techniques}
\label{sec:techniques-knapsack}

As in the classical knapsack problem with a linear objective, the algorithms achieving optimal approximation are based on enumeration techniques. One such approach is to enumerate the most valuable items in $\OPT$ (in the submodular problem, we can determine which items of $\OPT$ are valuable based on the Greedy ordering of $\OPT$, see \ref{eq1}) and greedily pack the remaining items based on the marginal gain to cost density.  This approach leads to the optimal $1-1/e$ approximation provided that we enumerate $3$ items \cite{Sviridenko04}. The running time of the resulting algorithm is $O(n^5)$ and it can be improved to $O(n^4 \log(n/\eps)/\eps)$ time at a loss of $\eps$ in the approximation. 

A different approach, inspired by the algorithms for the classical knapsack problem that use dynamic programming over the (appropriately discretized) profits of the items, is to enumerate over the marginal gains of the valuable items of $\OPT$. Unlike the classical setting with linear profits, it is considerably more challenging to leverage such an approach in the submodular setting. Badanidiyuru and Vondrak \cite{Badanidiyuru2014} propose a new approach based on this enumeration technique and continuous density Greedy with a running time of $n^2 \left({\log{n} \over \eps}\right)^{O\left({1 \over \eps^8} \right)}$, which overcame the $\Omega(n^4)$ running time barrier for the approaches that are based on enumerating items.

In this work, we build on the approach introduced by \cite{Badanidiyuru2014} and we obtain a faster running time of $\left({1 \over \eps}\right)^{O\left({1 \over \eps^4}\right)} n \log^2{n}$. Our algorithm is impractical due to the high dependency on $\eps$, but it is theoretically interesting. Obtaining near-optimal approximations in nearly-linear time for submodular maximization has been out of reach for all but a cardinality constraint.

Obtaining a fast running time poses several conceptual and technical challenges, and we highlight some of them here. Let us denote the valuable items of $\OPT$ as $\OPT_1$, and let $\OPT_2 = \OPT \setminus \OPT_1$. For our algorithm, the set $\OPT_1$ has $\mathrm{poly}(1/\eps)$ items and we can handle them by enumerating over their marginal gains, appropriately discretized. Similarly to \cite{Badanidiyuru2014}, we use the guessed marginal gains to pack items that are competitive with $\OPT_1$: for each guessed marginal gain, we find the cheapest item whose marginal gain is at least the guessed value, and we add $\eps$ of the item to the fractional solution. The continuous approach is necessary for ensuring that we obtain a good approximation, but it is already introducing the following conceptual and technical difficulties:
\begin{compactenum}
\item \emph{We do not know how much budget is available for the remaining items.} Since we packed the items fractionally, we will need to perform the rounding to find out which of the items will be in the final solution and their total budget. But we cannot do the rounding before packing the remaining items. Additionally, we cannot afford to guess the budget of $\OPT_1$, even approximately.
\item \emph{In the continuous setting, evaluating the multilinear extension takes $\Omega(n^2)$ time in general.}
\item \emph{We will need to ensure that we can round the resulting fractional solution.}
\end{compactenum}
A key idea in our algorithm, and an important departure from the approach of \cite{Badanidiyuru2014}, is to \emph{integrally} pack the remaining items using density Greedy with lazy evaluations to obtain a nearly-linear running time. The resulting fractional solution has only a constant number of entries that are strictly fractional, and we show that this is beneficial both in terms of running time and rounding: we can evaluate the multilinear extension in constant time and we can exploit the special structure of the solution to round. However, the first difficulty mentioned above remains a significant conceptual barrier for realizing this plan: if we cannot get a handle on how much budget to allocate to density Greedy, we will not be able to round the solution without violating the budget or losing value. Our solution here is based on the following insights.

First, note that we may assume that every item in $\OPT_2$ has a cost that is small relative to the total budget of $\OPT_2$: there can only be a small number of heavy items and each of them has small marginal gain on top of $\OPT_1$, and thus we can discard them without losing too much in the approximation. Moreover, if there are no heavy items at all, we can show that density Greedy will not exceed the budget. Thus, if we knew the budget of $\OPT_2$, we could remove all of the heavy items and run density Greedy on the remaining items.

Unfortunately, we cannot guess the budget of $\OPT_2$ since there are too many possible choices. Instead, note that, since the cost of an item is its marginal value divided by its density, a heavy item has large value or small density. If it has small density then intuitively Greedy will not pick it. The problematic items are the ones that have large marginal values, as density Greedy may pick them and they may be too heavy. Unfortunately, we cannot filter out all the items with large marginal value, since those items may include items in $\OPT_2$ (note that even though every item in $\OPT_2$ has small marginal value on top of $\OPT_1$, it can have large marginal value on top of our current fractional solution that does not necessarily contain $\OPT_1$). Now the key observation is that the number of such items is small, and we can handle them with additional guessing.

The final step of the algorithm is to round the fractional solution to a feasible integral solution. Here we take advantage of the fact that the only entries that are strictly fractional were introduced in the $\OPT_1$ stages of the algorithm. The fractional items can be mapped to the items in $\OPT_1$ in such a way that every item in $\OPT_1$ is assigned a fractional mass of at most $1$ coming from items with smaller or equal cost. Thus, for each item in $\OPT_1$, we want to select one of the items fractionally assigned to it. This is reminiscent of a partition matroid and thus a natural approach is to use a matroid rounding algorithm such as pipage rounding or swap rounding. However, an item may be fractionally assigned to more than one item in $\OPT_1$, and we need to ensure that the rounding does not select the same item for different items in $\OPT_1$. We show that we can do so using a careful application of swap rounding.  

\begin{algorithm}[t]
\begin{algorithmic}[1]
\State $t \gets 1 / \eps^3$
\State $r \gets 1/\eps$
\State $M \gets \Theta(f(\OPT))$
\State $S_{\mathrm{best}} \gets \emptyset$
\State Try all possible sequences:
\State \quad $\{v_{p, i}\}$: $p \in \{1, 2, \ldots, 1/\eps\}$, $i \in \{1, 2, \ldots, t\}$, $v_{p, i} \in \{0, \eps M/t, 2 \eps M/t, \ldots, 1\}$
\State \quad $\{W_p\}$: $p \in \{1, 2, \ldots, 1/\eps\}$, $W_p \in \{0, \eps M, 2 \eps M, \ldots, M\}$
\State \quad $\{w_{p, i}\}$: $p \in \{1, 2, \ldots, 1/\eps\}$, $i \in \{1, 2, \ldots, r + 1\}$, $w_{p, i} \in \{0, \eps^2 W_p / r, 2 \eps^2 W_p / r, \ldots, W_p \}$ 
\For{every choice $\{v_{p, i}\}$, $\{W_p\}$, $\{w_{p, i}\}$}
    \State $x \gets \textproc{KnapsackGuess}(f, \epsilon, \{v_{p, i}\}, \{W_p\}, \{w_{p, i}\})$
    \State $S \gets \textproc{Round}(x)$
    \If{$f(S) > f(S_{\mathrm{best}})$}
      \State $S_{\mathrm{best}} \gets S$
    \EndIf
\EndFor
\State Return $S_{\mathrm{best}}$
\end{algorithmic}
\caption{$\textproc{Knapsack}(f, \epsilon)$}
\label{alg:monotone}
\end{algorithm}

\section{The algorithm}
\label{sec:knapsack-algo}

We consider the problem of maximizing a monotone submodular function subject to a single knapsack constraint. Each element $e \in V$ has a cost $c_e \in \mathbb{R}_+$, and the goal is to find a set $\OPT \in \argmax\{f(S) \colon \sum_{e \in S} c_e \leq 1\}$. We assume that the knapsack capacity is $1$, which we may assume without loss of generality by scaling the cost of each element by the knapsack capacity. We also assume without loss of generality that $f(\emptyset) = 0$.

We let $F: [0, 1]^V \rightarrow \mathbb{R}_+$ denote the multilinear extension $f$. For every $x \in [0, 1]^V$, we have
  \[ F(x) = \sum_{S \subseteq V} f(S) \prod_{e \in S} x_e \prod_{e \notin S} (1 - x_e) = \Ex[R(x)],\]
where $R(x)$ is a random set that includes each element $e \in V$ independently with probability $x_e$.

We fix an optimal solution to the problem that we denote by $\OPT$. We assume that the algorithm knows a constant approximation of $f(\OPT)$; such an approximation can be obtained in nearly linear time by tacking the best of the following two solutions: the solution obtained by running Density Greedy (implemented using lazy evaluations, similarly to Algorithm~\ref{alg:lazy-density-greedy}) and the solution consisting of the best single element. Let $f(\OPT) \ge M \ge (1-\eps)f(\OPT)$ denote the algorithm's guess for the optimal value. There are $O(1/\eps)$ choices for $M$ given the constant approximation of $f(\OPT)$.

We consider the following Greedy ordering of $\OPT$. We order $\OPT$ as $o_1, o_2, \dots, o_{|\OPT|}$, where
\begin{equation}
\label{eq1}
  o_i = \argmax_{o \in \OPT} (f(\{o_1,\ldots, o_{i-1}\} \cup \{o\}) - f(\{o_1, \ldots, o_{i - 1}\}))
\end{equation}
Let $t = O(1/\eps^3)$, $\OPT_1 = \{o_1, o_2, \dots, o_t\}$, and $\OPT_2 = \OPT \setminus \OPT_1$.

We emphasize that we use the above ordering of $\OPT$ and the partition of $\OPT$ into $\OPT_1$ and $\OPT_2$ only for the analysis and to motivate the choices of the algorithm. In particular, the algorithm does not know this ordering or partition.

It is useful to filter out from $\OPT_2$ the items that have large cost, more precisely, cost greater than $\eps^2 (1 - c(\OPT_1))$. Since every element $o \in \OPT_2$ satisfies $f(\OPT_1 \cup \{o\}) - f(\OPT_1) \leq \eps^3 f(\OPT_1)$ and there are at most $1/\eps^2$ such elements, this will lead to only an $\eps f(\OPT)$ loss (see Appendix~\ref{app:omitted}). For ease of notation, we use $\OPT_2$ to denote the set without these elements, i.e., we assume that $c_o \leq \eps^2 (1 - c(\OPT_1))$ for every $o \in \OPT_2$. 

Algorithm~\ref{alg:monotone} gives a precise description of the algorithm. The algorithm guesses a sequence of values as follows.

{\bf Guessed values.}  Throughout the paper, we assume for simplicity that $1/\eps$ is an integer. Recall that $t = 1/\eps^3$. Let $r = 1/\eps$ ($r$ is an upper bound on the number of items of $\OPT_2$ that have large marginal value in each phase).
\begin{compactitem}
  \item A sequence $\left\{v_{1,1}, v_{1,2}, \ldots, v_{{1 / \eps}, t}\right\}$ where $v_{p,i} \in \{0, \eps M/t, 2\eps M/t,\ldots, M\}$ is an integer multiple of $\eps M / t$, for all integers $p$ and $i$ such that $1 \leq p \leq 1/\eps$ and $1 \leq i \leq t$. The value $v_{p, i}$ is an approximate guess for the marginal value of $o_i \in \OPT_1$ during phase $p$. There are $t / \eps = 1/\eps^4$ choices for each $v_{p, i}$ and thus there are $(1/\eps^4)^{1/\eps^4} = (1/\eps)^{O(1/\eps^4)}$ possible sequences.
  \item A sequence $\left\{W_1, W_2, \ldots, W_{1/\eps} \right\}$ where $W_p \in \{0, \eps M, 2 \eps M, \ldots, M\}$ is an integer multiple of $\eps M$, for all integers $p$ such that $1 \leq p \leq 1/\eps$. The value $W_p$ is an approximate guess for the total marginal value of $\OPT_2$ in phase $p$. There are $1/\eps$ choices for each $W_p$ and thus there are $(1/\eps)^{1/\eps}$ possible sequences.
  \item A sequence $\left\{w_{1,1}, w_{1,2}, \ldots, w_{{1/\eps}, {1/\eps} + 1} \right\}$ where $w_{p, i} \in \{0, \eps^2 W_p / r, 2 \eps^2 W_p / r, \ldots, W_p \}$ is an integer multiple of $\eps^2 W_p / r$, for all integers $p$ and $i$ such that $1 \leq p, i \leq 1/\eps$ (the value $W_p$ is the same as in the sequence above). The values $w_{p, i}$, where $i \in \{1, 2, \ldots, 1/\eps\}$, are approximate guesses for the marginal values of the items in $\OPT_2$ with large marginal value in phase $p$. There are $r/(\eps^2 + \eps) = 1/(\eps^3 + \eps^2)$ choices for each $w_{p, i}$ and thus there are $(1/(\eps^3 + \eps^2))^{1/\eps^2} = (1/\eps)^{O(1/\eps^2)}$ possible sequences. 
\end{compactitem}

The algorithm enumerates all possible such sequences. For each choice, the algorithm works as follows. Let $\{v_{p, i}\}$, $\{W_p\}$, and $\{w_{p, i}\}$ denote the current sequences. The algorithm performs $1/\eps$ phases. Each phase is comprised of three stages, executed in sequence in this order: an $\OPT_1$ stage, a stage for the large value items in $\OPT_2$, and a Density Greedy stage. We describe each of these stages in turn.

\smallskip
{\bf The $\OPT_1$ stage of phase $p$.} This stage uses the values $\{v_{p, i} \colon 1 \leq i \leq t\}$ as follows. We perform $t$ iterations. In each iteration $i$, we consider the items not selected in previous iterations that have marginal value at least $v_{p, i}$ on top of the current solution, i.e., $F(x \vee \mathbf{1}_e) - F(x) \geq v_{p, i}$. Among these items, we select the item with minimum cost and increase its fractional value by $\eps$. Together, the $t$ iterations select $t$ different items and increase their fractional value by $\eps$. 

\smallskip
{\bf The stage of phase $p$ for the large value items in $\OPT_2$.} This stage uses the value $W_p$ and the values $\{w_{p, i} \colon 1 \leq i \leq 1/\eps\}$ as follows. We perform at most $r$ iterations. In each iteration $i$, we find the minimum cost element that has marginal value at least $w_{p, i}$ on top of the current solution, and we integrally select this item. (Note that this is similar to the $\OPT_1$ stage, except that we select items integrally.) At the end of the stage, if the items selected in this phase have total marginal gain at least $\eps (1 - 12 \eps) W_p$, then we end phase $p$ and proceed to the next phase. Otherwise, the algorithm proceeds to the Density Greedy stage. 

\smallskip
{\bf The Density Greedy stage of phase $p$.}
If the previous stage did not reach a total marginal gain of at least $\eps (1 - 12 \eps) W_p$, we run the discrete Density Greedy algorithm until we reach a gain of $\eps (1 - 12 \eps) W_p$. Before running Density Greedy, we remove from consideration all elements whose marginal value is at least $\eps W_p/r$. In every step, the Density Greedy algorithm fully selects the item with largest density, i.e., ratio of marginal value to cost.

In order to achieve nearly linear time, we implement the Density Greedy algorithm using approximate lazy evaluations as shown in Algorithm~\ref{alg:lazy-density-greedy}. We maintain the items in a priority queue sorted by density. We initialize the marginal values and the densities with respect to the initial solution. In each iteration of the algorithm, we find an item whose density with respect to the current solution is within a factor of $(1 - \eps)$ of the maximum density as follows. We remove the item at the top of the queue. The marginal value of the item may be stale, so we evaluate its marginal gain with respect to the current solution. If the new marginal gain is within a factor of $(1 - \eps)$ of the old marginal gain, it follows from submodularity that the density of the item is within a factor of $(1 - \eps)$ of the maximum density, and we select the item. If the marginal gain has changed by a factor larger than $(1 - \eps)$, we update the density and reinsert the item in the queue. We also keep track of how many times each item's density has been updated and, if an item has been updated more than $2\ln(n/\eps) / \eps$ times, we discard the item since it can no longer contribute a significant value to the solution.

\smallskip
{\bf Rounding the fractional solution.}
After $1/\eps$ phases, we obtain a fractional solution with $O(1/\eps^4)$ fractional entries. We round the resulting fractional solution to an integral solution using swap rounding, as shown in Algorithm~\ref{alg:round}.


\begin{algorithm}
\begin{algorithmic}[1]
\State $t \gets 1/\eps^3$
\State $r \gets 1/\eps$
\State $x_0 \gets 0$ 
\For{$p = 1, 2, \dots, 1 / \eps$}
  \State $y^{(p, 0)} \gets x_{p - 1}$
  \State $A_p \gets \emptyset$
  \For{$i = 1, 2, \dots, t$}
    \State $a_{p, i} \gets \text{ element with minimum size $c_e$ in } \{e \notin A_p \colon F(y^{(p, i - 1)} \vee \mathbf{1}_{e}) - F(y^{(p, i - 1)}) \geq v_{p, i} \}$
    \State $y^{(p, i)} \gets y^{(p, i-1)} + \eps \mathbf{1}_{a_{p, i}}$
    \State $A_p \gets A_p \cup \{a_{p, i}\}$
  \EndFor

  \If{$W_p = 0$}
    \State Continue to the next phase $p+1$
  \EndIf
  
  \State $z^{(p, 0)} \gets y^{(p, t)}$
  \State $B_p \gets \emptyset$
  \State Let $r_p$ be the smallest $i \in \{0, 1, \ldots, r\}$ such that $w_{p, i + 1} \le \eps (1 - \eps) W_p / r$. If no such $i$ exists, let $r_p = r$.
  \Comment{$r_p$ is the number of large value elements in $\OPT_2$}

  \For{$i = 1, 2, \ldots, r_p$}
    \State $b_{p, i} \gets \text{ element with minimum size $c_e$ in } \{e \colon F(z^{(p, i - 1)} \vee \mathbf{1}_{e}) - F(z^{(p, i - 1)}) \geq w_{p, i} \}$
    \State $z^{(p, i)} \gets z^{(p, i-1)} \vee \mathbf{1}_{b_{p, i}}$
    \State $B_p \gets B_p \cup \{b_{p, i}\}$
    \If{$F(z^{(p, i)}) - F(z^{(p,0)}) \ge \eps(1-12\eps)W_p$}
      \State Set $x_p \gets z^{(p,i)}$ and continue to phase $p+1$ 
    \EndIf
  \EndFor
  
  \If{$F(z^{(p, r_p)}) - F(z^{(p, 0)}) < \eps(1-12\eps)W_p$}
    \State $V' \gets V \setminus \{e \colon F(z^{(p, r_p)} \vee \mathbf{1}_{e}) - F(z^{(p, r_p)}) \ge \eps W_p/r \}$
    \State $C_p \gets \Call{DensityGreedy}{f, z^{(p, r_p)}, \eps(1-12\eps)W_p - F(z^{(p, r_p)}) + F(z^{(p, 0)}), V'}$
    \State $x_p \gets z^{(p, r_p)} \vee \mathbf{1}_{C_p}$
  \EndIf
\EndFor
  \State Return $x_{1/\eps}$
\end{algorithmic}
\caption{$\textproc{KnapsackGuess}(f, \epsilon, \{v_{p, i}\}, \{W_p\}, \{w_{p, i}\})$}
\label{alg:monotone-knapsack-guess}
\end{algorithm}

\begin{algorithm}[t]
\begin{algorithmic}[1]
\State $S_0 \gets \emptyset$
\State $D \gets \emptyset$
\State $u(e) \gets 0$ for all $e \in V'$
\State $v(e) \gets F(x \vee \mathbf{1}_e) - F(x)$ for all $e \in V'$
\State Maintain the elements in a priority queue sorted in decreasing order by key, where the key of each element $e$ is initialized to its density ${v(e) \over c(e)}$ 
\For{$i = 1, 2, \dots$}
  \While{true}
    \If{queue is empty}
      \State \Return $S_{i-1}$
    \EndIf

    \State Remove the element $e$ from the priority queue with maximum key
    \State $v'(e) \gets F(x \vee \mathbf{1}_{S_{i - 1} \cup \{e\}}) - F(x \vee \mathbf{1}_{S_{i - 1}})$
    \State $u(e) \gets u(e) + 1$

    \If{$v(e) \geq (1 - \eps) v'(e)$}
      \State $e_i \gets e$
      \State $v(e) \gets v'(e)$
      
      \State $S_i \gets S_{i - 1} \cup \{e_i\}$
      \If{$f(x\vee 1_{S_i}) - f(x) \ge W$}
        \State \Return $S_i$
      \EndIf
      \State Exit the while loop and continue to iteration $i + 1$
    \Else
        \If{$u(e) \leq {2\ln(n/\eps) \over \eps}$}
          \State $v(e) \gets v'(e)$
          \State Reinsert $e$ into the queue with key ${v'(e) \over c(e)}$
        \Else
          \State $D \gets D \cup \{e\}$
        \EndIf
    \EndIf
  \EndWhile
\EndFor
\end{algorithmic}
\caption{$\textproc{LazyDensityGreedy}(f, x, W, V')$}
\label{alg:lazy-density-greedy}
\end{algorithm}

\clearpage


\section{Analysis of the running time}

Since all the fractional solutions considered have only $O(t/\eps)$ coordinates that strictly fractional, we can compute all marginal values exactly. Each evaluation of the multilinear extension takes $2^{O(t/\eps)} = 2^{O(1/\eps^4)}$ queries to the value oracle of $f$. 

Consider a single run of \textproc{KnapsackGuess}.The $\OPT_1$ stage of a phase (lines~{7}--{11}) computes $O(n t) = O(n/\eps^3)$ marginal values. The $\OPT_2$ guessing stage of a phase (lines~{18}--{25}) computes $O(n r) = O(n/\eps)$ marginal values. The filtering on line~{27} computes $O(n)$ marginal values. In the \textproc{LazyDensityGreedy} algorithm, the total number of queue operations is $O(n \log(n/\eps))$. Therefore \textproc{LazyDensityGreedy} computes $O(n \log(n/\eps))$ marginal values and $O(n \log(n/\eps) \log{n})$ additional time (removing an element from the queue takes $O(\log{n})$ time). Thus a phase of \textproc{KnapsackGuess} uses $2^{O(1/\eps^4)} n (\log(n / \eps) + {1 / \eps^3}) = 2^{O(1/\eps^4)} n \log{n}$ function evaluations and spends $O(n \log^2{n} \log(1/\eps))$ additional time. Since there are $1/\eps$ phases, \textproc{KnapsackGuess} uses ${1 \over \eps} \cdot 2^{O(1/\eps^4)} n \log{n} = 2^{O(1/\eps^4)} n \log{n}$ function evaluations and $O(n/\eps \log^2{n} \log(1/\eps))$ additional time.

The rounding algorithm \textproc{Round} uses $O(n + \log(1/\eps) / \eps^4)$ time and does not evaluate $f$.

Therefore, for each choice of the guessed values, \textproc{Knapsack} evaluates $f$ $2^{O(1/\eps^4)} n \log{n}$ times and spends $O(n/\eps \log^2{n} \log(1/\eps) + \log(1/\eps) /\eps^4)$ additional time. As discussed in Section~\ref{sec:knapsack-algo}, there are $(1/\eps)^{O(1/\eps^4)}$ possible choices for the guessed values. Therefore overall the algorithm uses $(1/\eps)^{O(1/\eps^4)} n \log{n}$ evaluation queries and $(1/\eps)^{O(1/\eps^4)} n \log^2{n}$ additional time.

\section{Analysis of the fractional solution}
\label{sec:fractional}
In this section, we prove the following theorem. In Section~\ref{sec:rounding}, we will use the second guarantee in the theorem statement in order to round the fractional solution without any loss.

\begin{theorem}
\label{thm:monotone-fractional}
  There are choices for the guessed values $\{v_{p, i}\}$, $\{W_p\}$, and $\{w_{p, i}\}$ for which Algorithm~\ref{alg:monotone-knapsack-guess} returns a fractional solution $x$ with the following properties:
  \begin{enumerate}[$(1)$]
    \item $F(x) \geq \left(1 - {1 \over e} - O(\eps) \right) f(\OPT)$;
    \item Let $E$ be the set of all items $e \in V$ such that $0 < x_e < 1$. There exists a mapping $\sigma: E \times \{1, 2, \dots, 1/\eps\} \rightarrow \OPT_1$ with the following properties:
    \begin{enumerate}[$(a)$]
      \item For every element $e \in E$ and every phase $p \in \{1, 2, \dots, 1/\eps\}$ such that $e \in A_p$, $\sigma(e, p)$ is defined and $c(e) \leq c(\sigma(e, p))$.
      \item For every element $o \in \OPT_1$, there are at most $1/\eps$ pairs $(e, p)$ such that $\sigma(e, p) = o$. 
    \end{enumerate}
  \end{enumerate} 
\end{theorem}

In the following, we fix a phase $p$ of the algorithm, and we analyze each of the stages of the phase. In the following lemma, we analyze the $\OPT_1$ stage of phase $p$ (lines 5--11 of \textproc{KnapsackGuess}).

\begin{lemma}
\label{lem1}
  Consider phase $p$ of the algorithm. There exist choices for the guessed values $\{v_{p, i}\}$ for which we have
\begin{enumerate}[$(1)$]
\item $F(y^{(p, t)}) - F(y^{(p, 0)}) \geq \eps (F(y^{(p, t)} \vee \mathbf{1}_{\OPT_1}) - F(y^{(p, t)})) - \eps^2 M$, 
\item $\mathrm{sorted}(c_{a_{p,1}}, c_{a_{p,2}}, \ldots, c_{a_{p,t}}) \le \mathrm{sorted}(c_{o_1}, c_{o_2}, \dots, c_{o_t})$, where $\{o_1, o_2, \dots, o_t\} = \OPT_1$.
\end{enumerate}
\end{lemma}
\begin{proof}
  We recursively define the values $v_{p, 1}, v_{p, 2}, \dots, v_{p, t}$, and a permutation $o'_1, o'_2, \dots, o'_t$ of the elements in $\OPT_1$ as follows. Suppose we have already defined $v_{p, 1}, \dots, v_{p, i - 1}$ and $o'_1, \dots, o'_{i - 1}$.
Let
  \[ \tilde{o}_i = \argmax_{o\in \OPT_1\setminus \{o'_1, \ldots, o'_{i-1}\}} (F(y^{(p,i-1)} \vee \mathbf{1}_{o}) - F(y^{(p,i-1)})).\]
We define
\begin{align*}
    v_{p,i} &= (\eps M/t) \left\lfloor\frac{F(y^{(p,i-1)} \vee \mathbf{1}_{\tilde{o}_i}) - F(y^{(p,i-1)})}{\eps M/t} \right\rfloor\\
    o'_{i} &=
    \begin{cases}
      a_{p, i} &\quad \text{ if } a_{p, i} \in \OPT_1 \setminus \{o'_1, \ldots, o'_{i - 1}\}\\
      \tilde{o}_i &\quad \text{ otherwise}
    \end{cases}
\end{align*}
In the definition of $o'_i$ above, the element $a_{p, i}$ is the one chosen on line 8 of \textproc{KnapsackGuess} based on the value $v_{p, i}$ defined above.

Let us now verify that these values $v_{p, i}$ satisfy the properties in the statement of the lemma. We first show the second property. We can show that $o'_i$ is a candidate for $a_{p, i}$ as follows. This is trivially true if $o'_i = a_{p, i}$ and thus we may assume that $o'_i = \tilde{o}_i$. Since the marginal value of $\tilde{o}_i$ is at least $v_{p, i}$, it suffices to show that $\tilde{o}_i \notin \{a_{p, 1}, \dots, a_{p, i - 1}\}$. It is straightforward to verify by induction that, for all $j$, $\{a_{p, 1}, \dots, a_{p, j}\} \cap \OPT_1 \subseteq \{o'_1, \dots, o'_j\}$. Since $\tilde{o}_i \in \OPT_1 \setminus \{o'_1, \dots, o'_{i - 1}\}$, it follows that $\tilde{o}_i \notin \{a_{p, 1}, \dots, a_{p, i - 1}\}$. Therefore $o'_i$ is a candidate for $a_{p, i}$ and thus $c_{a_{p, i}} \leq c_{o'_i}$ for all $i$. Since the elements $o'_1, \dots, o'_t$ are a permutation of $\OPT_1$, we have $\mathrm{sorted}(c_{a_{p, 1}}, \dots, c_{a_{p, t}}) \leq \mathrm{sorted}(c_{o_1}, \dots, c_{o_t})$.

We now show the first property. For each $i$, we have
\begin{align*}
  F(y^{(p, i)}) - F(y^{(p, i-1)}) &= F(y^{(p, i - 1)} + \eps \mathbf{1}_{e_{p, i}}) - F(y^{(p, i - 1)})\\
  &= \eps (F(y^{(p, i - 1)} \vee \mathbf{1}_{e_{p, i}}) - F(y^{(p, i - 1)}))\\
  &\geq \eps v_{p, i}\\
  &= \eps {\eps M \over t} \left\lfloor {F(y^{(p, i - 1)} \vee \mathbf{1}_{\tilde{o}_i}) - F(y^{(p, i - 1)}) \over \eps M / t}\right\rfloor\\
  &\geq \eps {\eps M \over t} \left( {F(y^{(p, i - 1)} \vee \mathbf{1}_{\tilde{o}_i}) - F(y^{(p, i - 1)}) \over \eps M / t} - 1\right)\\
  &= \eps (F(y^{(p, i - 1)} \vee \mathbf{1}_{\tilde{o}_i}) - F(y^{(p, i - 1)})) - {\eps^2 M \over t}\\
  &= F(y^{(p, i - 1)} + \eps \mathbf{1}_{\tilde{o}_i}) - F(y^{(p, i - 1)}) - {\eps^2 M \over t}.
\end{align*}
Since $o'_i \in \{a_{p, i}, \tilde{o}_i\}$, it follows that
\begin{align*}
 F(y^{(p, i)}) - F(y^{(p, i - 1)}) &\geq F(y^{(p, i - 1)} + \eps \mathbf{1}_{o'_i}) - F(y^{(p, i - 1)}) - {\eps^2 M \over t}\\
  &= \eps (F(y^{(p, i - 1)} \vee \mathbf{1}_{o'_i}) - F(y^{(p, i - 1)})) - {\eps^2 M \over t}.
\end{align*}
By summing up all these inequalities and using submodularity, we obtain
\begin{align*}
  F(y^{(p, t)}) - F(y^{(p, 0)}) &\geq \eps \sum_{i = 1}^t (F(y^{(p, i - 1)} \vee \mathbf{1}_{o'_i}) - F(y^{(p, i - 1)})) - {\eps^2 M}\\
  &\geq \eps (F(y^{(p, t)} \vee \mathbf{1}_{\OPT_1}) - F(y^{(p, t)})) - \eps^2 M.
\end{align*}
\end{proof}

In the following lemma, we analyze the stage of phase $p$ for large value elements of $\OPT_2$ (lines 15--25 of \textproc{KnapsackGuess}). The proof is similar to the proof of Lemma~\ref{lem1}.

\begin{lemma}
\label{lem1forOPT2}
  There exist choices for the guessed values $W_p$ and $\{w_{p, i}\}$ for which we have 
\begin{enumerate}[$(1)$]
\item There is a subset $O =\{o'_1, \ldots, o'_{r_p}\} \subseteq \OPT_2$ such that $c_{b_{p, i}} \le c_{o'_i}$ for all $1 \leq i \leq r_p$.
\item Consider $i \in\{1, \ldots, r_p\}$. Let $O_i = \{o'_1,\ldots, o'_i\}$. We have
  \[ F(z^{(p, i)}) - F(z^{(p, 0)}) \geq F(z^{(p, i)} \vee \mathbf{1}_{O_i}) - F(z^{(p, i)}) - \eps^2 W_p,\]
  and
  \[F(z^{(p, i)}) - F(z^{(p, 0)}) \geq \frac{(1-5\eps)c(O_i)(F(z^{(p,i)}\vee \mathbf{1}_{\OPT_2}) - F(z^{(p,i)}))}{c(\OPT_2)} - \eps^2 W_p.\]
Furthermore, if the phase does not end after iteration $i$ (line 23 of \textproc{KnapsackGuess}) then
$$F(z^{(p, i)}) - F(z^{(p, 0)}) \ge \frac{(1-6\eps) c(O_i)}{c(\OPT_2)}W_p - \eps^2 W_p$$
$$c(O_i) \le \eps(1-4\eps) c(\OPT_2)$$
\item Consider $i \in\{1, \ldots, r_p\}$. Let $O_i = \{o'_1,\ldots, o'_i\}$. For every $o \in \OPT_2\setminus O_i$, at least one of the following conditions holds:
\begin{itemize}
\item $F(z^{(p, i)} \vee \mathbf{1}_{o}) - F(z^{(p, i)}) \leq w_{p, i + 1} + \eps^2 W_p/r$
\item ${F(z^{(p, i)} \vee \mathbf{1}_{o}) - F(z^{(p, i)}) \over c_o} < {(1-5\eps) (F(z^{(p, i)} \vee \mathbf{1}_{\OPT_2}) - F(z^{(p, i)})) \over c(\OPT_2)}$
\end{itemize}
\end{enumerate}
\end{lemma}
\begin{proof}
We define
  \[ W_p = (\eps M) \left\lfloor {F(z^{(p, 0)} \vee \mathbf{1}_{\OPT_2}) - F(z^{(p, 0)}) \over \eps M} \right\rfloor.\]
We define the values $w_{p, 1}, \ldots, w_{p, r_p}$ and a sequence of distinct elements $o'_1, \ldots, o'_{r_p}$ in $\OPT_2$ recursively as follows. Suppose we have already defined the values $w_{p, 1}, \ldots, w_{p, i - 1}$ and the elements $o'_1, \ldots, o'_{i - 1}$. Let $\mathrm{SO}_i$ be the set of elements $o \in \OPT_2 \setminus \{o'_1, \ldots, o'_{i - 1}\}$ that satisfy
  \[ {F(z^{(p,i-1)} \vee \mathbf{1}_{o}) - F(z^{(p,i-1)}) \over c_o} \geq {(1 - 5\eps) (F(z^{(p, i-1)} \vee \mathbf{1}_{\OPT_2}) - F(z^{(p, i-1)})) \over c(\OPT_2)} \]
Let
  \[ \tilde{o}_i = \argmax_{o \in \mathrm{SO}_i} (F(z^{(p, i - 1)} \vee \mathbf{1}_{o}) - F(z^{(p,i-1)})).\]
We define
\begin{align*}
  w_{p,i} &= (\eps^2 W_p/r) \left\lfloor {F(z^{(p,i-1)} \vee \mathbf{1}_{\tilde{o}_i}) - F(z^{(p,i-1)}) \over \eps^2 W_p/r} \right\rfloor\\
  o'_i &= 
  \begin{cases}
    b_{p, i} & \text{ if } b_{p, i} \in \mathrm{SO}_i\\
    \tilde{o}_i & \text{ otherwise}
  \end{cases}
\end{align*}
In the definition of $o'_i$ above, the element $b_{p, i}$ is the one chosen on line 19 of \textproc{KnapsackGuess} based on the values $W_p$ and $w_{p, i}$ defined above.

We now verify that the values $W_p$ and $\{w_{p, i}\}$ satisfy the properties in the statement of the lemma. Let $O = \{o'_1, \ldots, o'_{r_p}\}$ and $O_i = \{o'_1, \ldots, o'_i\}$ for all $1 \leq i \leq r_p$. 

The first property follows from the fact that $o'_i$ is a candidate for $b_{p, i}$.

We next show the third property, which follows from the definition of $w_{p, i + 1}$. If $o \notin \mathrm{SO}_{i + 1}$ then the second condition holds by definition of $\mathrm{SO}_{i + 1}$. Therefore we may assume that $o \in \mathrm{SO}_{i + 1}$. By the definition of $w_{p, i + 1}$ and $\tilde{o}_{i + 1}$, we have
\begin{align*}
  w_{p, i + 1} &= {\eps^2 W_p \over r} \left\lfloor {F(z^{(p,i)} \vee \mathbf{1}_{\tilde{o}_{i + 1}}) - F(z^{(p,i)}) \over \eps^2 W_p/r} \right\rfloor\\
  &\geq {\eps^2 W_p \over r} \left( {F(z^{(p,i)} \vee \mathbf{1}_{\tilde{o}_{i + 1}}) - F(z^{(p,i)}) \over \eps^2 W_p/r} - 1 \right)\\
  &= F(z^{(p,i)} \vee \mathbf{1}_{\tilde{o}_{i + 1}}) - F(z^{(p,i)}) - {\eps^2 W_p \over r}\\ 
  &\geq F(z^{(p,i)} \vee \mathbf{1}_{o}) - F(z^{(p,i)}) - {\eps^2 W_p \over r}. 
\end{align*}
By rearranging the inequality above, we obtain that $o$ satisfies the first condition of property $(3)$.

We now show the second property. For each $i$, we have
\begin{align*}
  F(z^{(p, i)}) - F(z^{(p, i-1)}) &= F(z^{(p, i - 1)} \vee \mathbf{1}_{b_{p, i}}) - F(z^{(p, i - 1)})\\
  &\geq w_{p, i}\\
  &= {\eps^2 W_p \over r} \left\lfloor {F(z^{(p, i - 1)} \vee \mathbf{1}_{\tilde{o}_i}) - F(z^{(p, i - 1)}) \over \eps^2 W_p/r} \right\rfloor\\ 
  &\geq {\eps^2 W_p \over r} \left( {F(z^{(p, i - 1)} \vee \mathbf{1}_{\tilde{o}_i}) - F(z^{(p, i - 1)}) \over \eps^2 W_p / r} - 1\right)\\
  &= F(z^{(p, i - 1)} \vee \mathbf{1}_{\tilde{o}_i}) - F(z^{(p, i - 1)}) - {\eps^2 W_p \over r}
\end{align*}
Since $o'_i \in \{b_{p, i}, \tilde{o}_i\}$ and $o'_i \in \mathrm{SO}_i$, it follows that
\begin{align*}
  F(z^{(p, i)}) - F(z^{(p, i - 1)}) &\geq F(z^{(p, i - 1)} \vee \mathbf{1}_{o'_i}) - F(z^{(p, i - 1)}) - {\eps^2 W_p \over r}\\
 &\geq {(1 - 5\eps) c(o'_i) (F(z^{(p, i - 1)} \vee \mathbf{1}_{\OPT_2}) - F(z^{(p, i - 1)})) \over c(\OPT_2)} - {\eps^2 W_p \over r}
\end{align*}

By adding these inequalities for the first $i$ iterations and using submodularity, we obtain
\begin{align*}
  F(z^{(p, i)}) - F(z^{(p, 0)})
  &\geq \sum_{j = 1}^i (F(z^{(p, j - 1)} \vee \mathbf{1}_{o'_j}) - F(z^{(p, j - 1)})) - {i \eps^2 W_p \over r}\\
  &\geq F(z^{(p, i)} \vee \mathbf{1}_{O_i}) - F(z^{(p, i)}) - {i \eps^2 W_p \over r}\\
  &\geq F(z^{(p, i)} \vee \mathbf{1}_{O_i}) - F(z^{(p, i)}) - \eps^2 W_p. 
\end{align*}
Similarly,
\begin{align*}
 F(z^{(p, i)}) - F(z^{(p, 0)})
 &\geq \sum_{j = 1}^i {(1 - 5\eps) c(o'_j) (F(z^{(p, j - 1)} \vee \mathbf{1}_{\OPT_2}) - F(z^{(p, j - 1)})) \over c(\OPT_2)} - {i \eps^2 W_p \over r}\\
 &\geq \sum_{j = 1}^i {(1 - 5\eps) c(o'_j) (F(z^{(p, i)} \vee \mathbf{1}_{\OPT_2}) - F(z^{(p, i)})) \over c(\OPT_2)} - {i \eps^2 W_p \over r}\\
 &= {(1 - 5\eps) c(O_i) (F(z^{(p, i)} \vee \mathbf{1}_{\OPT_2}) - F(z^{(p, i)})) \over c(\OPT_2)} - {i \eps^2 W_p \over r}\\
  &\geq {(1 - 5\eps) c(O_i) (F(z^{(p, i)} \vee \mathbf{1}_{\OPT_2}) - F(z^{(p, i)})) \over c(\OPT_2)} - \eps^2 W_p
\end{align*}

If the phase does not end in iteration $i$ then $F(z^{(p,i)}) - F(z^{(p,0)}) \leq \eps(1 - 12\eps) W_p$. It follows that 
  \[ {(1-5\eps) c(O_i) (F(z^{(p, i)} \vee \mathbf{1}_{\OPT_2}) - F(z^{(p,i)})) \over c(\OPT_2)} - {i \eps^2 W_p \over r} \leq F(z^{(p,i)}) - F(z^{(p,0)}) \leq \eps(1 - 12\eps) W_p.\]
Additionally,
\begin{align*}
  F(z^{(p, i)} \vee \mathbf{1}_{\OPT_2}) - F(z^{(p, i)}) &\geq F(z^{(p, 0)} \vee \mathbf{1}_{\OPT_2}) - F(z^{(p, i)})\\
  &= (F(z^{(p, 0)} \vee \mathbf{1}_{\OPT_2}) - F(z^{(p, 0)})) - (F(z^{(p, i)}) - F(z^{(p, 0)}))\\
  &\geq W_p - \eps (1 - 12\eps) W_p.
\end{align*}
The first line follows from monotonicity. The third line follows from the definition of $W_p$ and the fact that $F(z^{(p, i)}) - F(z^{(p, 0)}) \leq \eps (1 - 12\eps) W_p$.

Therefore
  \[ {(1 - 5\eps) c(O_i) (W_p - \eps(1 - 12\eps)W_p) \over c(\OPT_2)} - {i \eps^2 W_p \over r} \leq F(z^{(p,i)}) - F(z^{(p,0)}) \leq \eps(1-12\eps)W_p.\]
Thus
  \[ F(z^{(p,i)}) - F(z^{(p,0)}) \geq {(1 - 6\eps) c(O_i)W _p \over c(\OPT_2)} - {i \eps^2 W_p \over r}, \]
and
  \[ c(O_i) \leq \eps(1 - 4\eps) c(\OPT_2).\]
\end{proof}

In the following lemma, we wrap up the analysis of phase $p$. After the $\OPT_1$ stage and the stage for the large value items in $\OPT_2$, either the phase ends because we have already collected the target marginal value or we use Density Greedy to collect the remaining value. In each of these cases, we show that we reach the target value of $\eps (1 - 12\eps) W_p$ and the total cost of the items we select is at most $\eps (1 - c(\OPT_1))$.

\begin{lemma}
\label{lem2}
  Suppose that we run the $\textproc{KnapsackGuess}$ algorithm with the values $\{v_{p, i}\}$, $\{w_{p, i}\}$, and $W_p$ guaranteed by Lemmas~\ref{lem1} and \ref{lem1forOPT2} as input. We have
\begin{itemize}
  \item $F(x_p) - F(z^{(p, 0)}) \geq \eps (1 - 12 \eps) W_p$, and 
  \item $c(B_p) + c(C_p) \le \eps(1-c(\OPT_1))$.
\end{itemize}
\end{lemma}
\begin{proof}
  We first consider the case when phase $p$ ends before running \textproc{LazyDensityGreedy} (on line~{23} of \textproc{KnapsackGuess}). We show that the lemma follows from Lemma~\ref{lem1forOPT2}. Since the first condition follows immediately from the fact that the phase ends on line~{23}, it suffices to verify the second condition. Since we do not run \textproc{LazyDensityGreedy}, we have $C_p = \emptyset$ and thus it suffices to show that $c(B_p) \leq \eps (1 - c(\OPT_1))$. Consider the last iteration $i$ where $F(z^{(p, i)}) - F(z^{(p, 0)}) < \eps(1 - 12\eps) W_p$. By property 1 of Lemma~\ref{lem1forOPT2}, we have $c(B_p) \leq c(O_{i + 1})$. By property 2 of Lemma~\ref{lem1forOPT2}, $c(O_i) \leq \eps (1 - 4\eps) c(\OPT_2)$. Additionally, $c_{o'_{i + 1}} \leq \eps^2 (1 - c(\OPT_1))$, since every item in $\OPT_2$ has cost at most $\eps^2 (1 - c(\OPT_1))$. Using these observations and the fact that $c(\OPT_2) \leq 1 - c(\OPT_1)$, we obtain
    \[ c(B_p) \leq c(O_{i + 1}) \leq \eps(1 - \eps) c(\OPT_2) + \eps^2 (1 - c(\OPT_1)) \leq \eps (1 - c(\OPT_1)).\]
  Next, we consider the case when \textproc{LazyDensityGreedy} is called. To simplify notation, in the remainder of the proof we use $x$ to denote the starting solution of \textproc{LazyDensityGreedy}, i.e., $x = z^{(p, r_p)}$.

  The following claim shows that, since the guessing stage for $\OPT_2$ did not pick up enough value, the marginal value of $\OPT_2 \setminus O$ on top of $x$ is large. Recall that $W_p$ is approximately the marginal value of $\OPT_2$ on top of $z^{(p, 0)}$, and thus the following claim is showing that $\OPT_2 \setminus O$ accounts for most of this total value.

  \begin{claim}
  \label{claim1}
    $F(x \vee \mathbf{1}_{\OPT_2 \setminus O}) - F(x)\geq (1 - 2\eps) W_p.$
  \end{claim}
  \begin{proof}
    Property~{2} in Lemma~\ref{lem1forOPT2} and the fact that the guessing stage for $\OPT_2$ does not pick up enough value give us the following inequalities:
    \begin{align*}
      F(x \vee \mathbf{1}_{O}) - F(x) &\leq F(x) - F(z^{(p, 0)}) + \eps^2 W_p\\
      F(x) - F(z^{(p, 0)}) &\leq \eps (1 - 12 \eps) W_p
    \end{align*}
    By adding the two inequalities, we obtain
    \begin{equation}
    \label{eq2}
      F(x \vee \mathbf{1}_O) - F(z^{(p, 0)})
      \leq F(x) - F(z^{(p, 0)}) + \eps (1 - 11\eps) W_p
      \leq 2 \eps W_p
    \end{equation}
    We have
    \begin{align*}
      F(x \vee \mathbf{1}_{\OPT_2 \setminus O}) - F(x) &\geq F(x \vee \mathbf{1}_{\OPT_2}) - F(x \vee \mathbf{1}_O)\\
      &= (F(x \vee \mathbf{1}_{\OPT_2}) - F(z^{(p, 0)})) - (F(x \vee \mathbf{1}_O) - F(z^{(p, 0)}))\\
      &\geq F(x \vee \mathbf{1}_{\OPT_2}) - F(z^{(p, 0)}) - 2 \eps W_p\\
      &\geq F(z^{(p, 0)} \vee \mathbf{1}_{\OPT_2}) - F(z^{(p, 0)}) - 2\eps W_p\\
      &\geq W_p - 2\eps W_p\\
      &\geq (1 - 2\eps) W_p. 
    \end{align*}
    On the first line, we used submodularity. On the third line, we used (\ref{eq2}). The fourth line follows from monotonicity. The fifth line follows from the definition of $W_p$.
  \end{proof}

  We now show that the filtering of the items right before we ran \textproc{LazyDensityGreedy} (line~{27} of \textproc{KnapsackGuess}) did not remove any element of $\OPT_2 \setminus O$ that would have been selected by \textproc{LazyDensityGreedy}. Let
  \[ L := \frac{(1-5\eps)(F(x \vee \mathbf{1}_{\OPT_2}) - F(x))}{c(\OPT_2)}. \]
  We will show that every element $o \in \OPT_2 \setminus O$ has marginal value on top of $x$ of at most $\eps^2 W_p$ or it has density less than $L$, and that the elements selected by \textproc{LazyDensityGreedy} have density at least $L$.

  \begin{claim}
  \label{claim3}
    For every $o \in \OPT_2 \setminus O$, $F(x \vee \mathbf{1}_o) - F(x) \leq \eps^2 W_p$ or ${F(x \vee \mathbf{1}_o) - F(x) \over c_o} < L$.
  \end{claim}
  \begin{proof}
    If $r_p = r$ then $w_{p, i} \geq \eps (1 - \eps) W_p / r$ for all $1 \leq i \leq r$ and thus
      \[ F(z^{(p, r)}) - F(z^{(p, 0)}) \geq \sum_{i = 1}^r w_{p, i} \geq \eps(1 - \eps) W_p.\]
      In this case, the phase ends before running \textproc{LazyDensityGreedy}. Thus we must have $r_p < r$ and $w_{p, r_p + 1} \leq \eps (1 - \eps) W_p / r$. By property~{3} in Lemma~\ref{lem1forOPT2}, for every $o \in \OPT_2 \setminus O$, we have
    \[ F(x \vee \mathbf{1}_o) - F(x) \le w_{p, r_p + 1} + \eps^2 W_p/r \leq \eps^2 W_p,\] 
    or
    \[ {F(x \vee \mathbf{1}_{o}) - F(x) \over c_o} < L.\]
  \end{proof}

Before showing that the algorithm stops before reaching density $L$, let us first address the elements that are removed from the queue on line~{27} (they are added to the set $D$ consisting of all elements that were updated too many times).  The following claims shows that their marginal values is negligible.

 \begin{claim}
 \label{claim4}
  Consider an iteration $i$ of \textproc{LazyDensityGreedy}. For every element $e \in D$, we have
    \[ F(x \vee \mathbf{1}_{S_i \cup \{e\}}) - F(x \vee \mathbf{1}_{S_i}) \leq \left({\eps \over n}\right)^2 f(\OPT).\]
 \end{claim}
 \begin{proof}
    Let $e \in D$ and suppose that $e$ was added to $D$ during iteration $j \leq i$. Then $e$ was updated more than $2\ln(n/\eps)/\eps$ times in the first $j$ iterations. Since each update happens when the marginal value decreases by at least a $(1 - \eps)$ factor, the marginal value of $e$ at the beginning of iteration $j$ is at most
    \[ (1 - \eps)^{{2\ln(n/\eps) \over \eps}} (F(x \vee \mathbf{1}_{e}) - F(x)) \leq \left({\eps \over n}\right)^2 (F(x \vee \mathbf{1}_{e}) - F(x)) \leq \left({\eps \over n}\right)^2 f(\OPT).\]
    The first inequality follows from the inequality $1 - x \leq e^{-x}$, and the second inequality follows from $f(\OPT) \geq \max_e f(\{e\})$.

    By submodularity, the marginal value of $e$ can only decrease between iteration $j$ and $i$, and the claim follows.
 \end{proof}

  \begin{claim}
  \label{claim5}
    Consider an iteration $i$ of \textproc{LazyDensityGreedy}. The density of the element $e_i$ selected in iteration $i$ is at least $L$, i.e., 
    \[ {F(x \vee \mathbf{1}_{S_i}) - F(x \vee \mathbf{1}_{S_{i - 1}}) \over c_{e_i}} \geq L.\]
  \end{claim}
  \begin{proof}
    Let $\OPT'_2 = \OPT_2 \setminus D$. Suppose that the density of every element $o \in \OPT'_2 \setminus (O \cup S_{i - 1})$ is less than $L / (1 - \eps)$, i.e.,
      \[ {F(x \vee \mathbf{1}_{S_{i - 1} \cup \{o\}}) - F(x \vee \mathbf{1}_{S_{i - 1}}) \over c_o} < {L \over 1 - \eps}.\]
    It follows that
    \begin{equation}
    \label{eq1}
      F(x \vee \mathbf{1}_{S_{i - 1} \cup (\OPT'_2 \setminus O)}) - F(x \vee \mathbf{1}_{S_{i - 1}}) < {L \cdot c(\OPT'_2 \setminus O) \over 1 - \eps}  \leq {L \cdot c(\OPT_2 \setminus O) \over 1 - \eps}.
    \end{equation}
  Using the above inequality, Claim~\ref{claim4}, and the facts that $F(x \vee \mathbf{1}_{\OPT_2}) - F(x) \geq W_p$ and $W_p \geq \eps M$ (if $W_p = 0$, we never run \textproc{LazyDensityGreedy}), we obtain
  \begin{align*}
    &F(x \vee \mathbf{1}_{S_{i - 1} \cup (\OPT_2 \setminus O)}) - F(x \vee  \mathbf{1}_{S_{i - 1}})\\ 
    &\leq F(x \vee \mathbf{1}_{S_{i - 1} \cup (\OPT'_2 \setminus O)}) - F(x \vee  \mathbf{1}_{S_{i - 1}}) + \eps^2 n^{-1} M\\
    &\leq {L c(\OPT_2 \setminus O) \over 1 - \eps} + \eps n^{-1} W_p\\
    &= {(1 - 5\eps) (F(x \vee \mathbf{1}_{\OPT_2}) - F(x)) c(\OPT_2 \setminus O) \over (1 - \eps) c(\OPT_2)} + \eps n^{-1} W_p\\
    &\leq {(1-5\eps) c(\OPT_2 \setminus O) W_p \over (1-\eps) c(\OPT_2)} + \eps n^{-1} W_p\\
    &\leq (1 - 3\eps) W_p.
  \end{align*}
  On the first line, we used Claim~\ref{claim4}. On the second line, we used (\ref{eq1}) and $W_p \geq \eps M$. On the third line, we used the definition of $L$. On the fourth line, we used $F(x \vee \mathbf{1}_{\OPT_2}) - F(x) \geq W_p$.

  Thus,
  \begin{align*}
    F(x\vee  \mathbf{1}_{S_{i - 1}}) - F(x) 
    &\geq -(1-3\eps) W_p +  F(x\vee \mathbf{1}_{S_i} \vee \mathbf{1}_{\OPT_2\setminus O}) - F(x)\\
    &\ge -(1-3\eps) W_p +  (1-2\eps) W_p\\
    &= \eps \cdot W_p
  \end{align*}
  In the second inequality, we used Claim~\ref{claim1}.

  Therefore the phase ends at the end of iteration $i - 1$, which is a contradiction. Thus some element in $\OPT'_2 \setminus O$ has density at least $L/(1 - \eps)$. Since $e_i$ has density at least $(1 - \eps)$ times the best density, it follows that the density of $e_i$ is at least $L$.
  \end{proof}

  Thus, Claims~\ref{claim3} and \ref{claim5} imply that all of the elements of $\OPT_2 \setminus O$ that are relevant for \textproc{LazyDensityGreedy} are included in $V'$. Now we can complete the proof as follows. As before, we let $\OPT'_2 = \OPT_2 \setminus D$, i.e., the subset of $\OPT_2$ that was not removed from the queue on line~{27}. 
  
  When the element $e_i$ is added, since every $o \in \OPT'_2 \setminus (O \cup S_{i - 1})$ is in the queue, we have
  \[
  \frac{F(x\vee \mathbf{1}_{S_i}) - F(x \vee \mathbf{1}_{S_{i-1}})}{c_{e_i}} \ge (1-\eps)\cdot\frac{F(x\vee \mathbf{1}_{S_{i-1}}\vee \mathbf{1}_o) - F(x\vee \mathbf{1}_{S_{i-1}})}{c_o}~\forall o\in \OPT'_2\setminus O.
  \]
  Thus,
  \[
  \frac{F(x \vee \mathbf{1}_{S_i}) - F(x \vee \mathbf{1}_{S_{i-1}})}{c_{e_i}} \ge (1-\eps)\cdot \frac{F(x \vee \mathbf{1}_{S_{i-1}} \vee \mathbf{1}_{\OPT'_2 \setminus O}) - F(x \vee \mathbf{1}_{S_{i-1}})}{c(\OPT'_2\setminus O)}.
  \]
  By summing up the above inequalities over all iterations $i \leq \ell$ and using submodularity, we obtain
  \[ F(x \vee \mathbf{1}_{S_{\ell}}) - F(x) \ge \frac{(1-\eps) c(S_{\ell})}{c(\OPT'_2\setminus O)} (F(x \vee \mathbf{1}_{S_{\ell}} \vee \mathbf{1}_{\OPT'_2\setminus O}) - F(x \vee \mathbf{1}_{S_{\ell}})).\] 
  If the algorithm does not terminate in iteration $\ell$ then $F(x \vee \mathbf{1}_{S_{\ell}}) - F(x) < \eps W_p$ and therefore
  \begin{align*}
    & F(x \vee \mathbf{1}_{S_{\ell}}) - F(x)\\
    &\geq \frac{(1-\eps)c(S_{\ell})}{c(\OPT'_2\setminus O)} (F(x \vee \mathbf{1}_{S_{\ell}} \vee \mathbf{1}_{\OPT'_2\setminus O}) - F(x) -\eps W_p)\\
    &= \frac{(1-\eps)c(S_{\ell})}{c(\OPT'_2\setminus O)} (F(x \vee \mathbf{1}_{S_{\ell}} \vee \mathbf{1}_{\OPT'_2\setminus O}) - F(x \vee \mathbf{1}_{S_{\ell}}) + F(x \vee \mathbf{1}_{S_{\ell}}) - F(x) - \eps W_p)\\
    &\geq \frac{(1-\eps)c(S_{\ell})}{c(\OPT'_2\setminus O)} (F(x \vee \mathbf{1}_{S_{\ell}} \vee \mathbf{1}_{\OPT_2\setminus O}) - F(x \vee \mathbf{1}_{S_{\ell}}) - \eps^2 n^{-1}M + F(x \vee \mathbf{1}_{S_{\ell}}) - F(x)  - \eps W_p)\\
    &\geq \frac{(1-\eps)c(S_{\ell})}{c(\OPT'_2\setminus O)} (F(x \vee \mathbf{1}_{S_{\ell}} \vee \mathbf{1}_{\OPT_2\setminus O}) - F(x) - 2\eps W_p)\\
    &\geq \frac{(1-\eps)c(S_{\ell})}{c(\OPT'_2\setminus O)} (F(x \vee \mathbf{1}_{\OPT_2\setminus O}) - F(x) - 2\eps W_p)\\
    &\geq \frac{(1-\eps)c(S_{\ell})}{c(\OPT'_2\setminus O)} (1 - 4\eps) W_p\\ 
    &\geq \frac{(1-\eps)c(S_{\ell})}{c(\OPT_2)} (1 - 4\eps) W_p. 
  \end{align*}
  On the third line, we have used Claim~\ref{claim4}. On the fourth line, we have used that $W_p \geq \eps M$. On the fifth line, we have used monotonicity. On the sixth line, we have used Claim~\ref{claim1}.

  Additionally, by property~{2} in Lemma~\ref{lem1forOPT2}, we have
    \[ F(x) - F(z^{(p, 0)}) \geq {(1 - 6 \eps) c(O) \over c(\OPT_2)} W_p - \eps^2 W_p.\]
  Thus, for any iteration $\ell$ where \textproc{LazyDensityGreedy} does not stop, we have
  \begin{align*}
    \eps (1 - 12\eps) W_p &\geq F(x \vee \mathbf{1}_{S_{\ell}}) - F(z^{(p, 0)})\\
    &= F(x \vee \mathbf{1}_{S_{\ell}}) - F(x) + F(x) - F(z^{(p, 0)})\\
    &\geq {(1-\eps)c(S_{\ell}) \over c(\OPT_2)} (1 - 4\eps) W_p + {(1 - 6 \eps) c(O) \over c(\OPT_2)} W_p - \eps^2 W_p\\
    &\geq {(1 - 6\eps) c(S_{\ell}) \over c(\OPT_2)} W_p + {(1 - 6 \eps) c(O) \over c(\OPT_2)} W_p - \eps^2 W_p. 
  \end{align*}
  By property~{1} in Lemma~\ref{lem1forOPT2}, $c(B_p) \leq c(O)$. Therefore
  \begin{align*}
    c(S_{\ell}) + c(B_p) &\leq c(S_{\ell}) + c(O)\\
    &\leq {\eps (1 - 11 \eps) \over 1 - 6\eps} c(\OPT_2)\\
    &\leq \eps (1 - 5\eps) c(\OPT_2).
  \end{align*}
  Finally, consider the last element $e$ selected by \textproc{LazyDensityGreedy}. By Claim~\ref{claim5}, the density of $e$ is at least $L \geq (1 - 5\eps) W_p / c(\OPT_2)$. Additionally, the marginal value of $e$ is at most $\eps^2 W_p$, since $e \in V'$. Therefore $c_e \le \eps^2 c(\OPT_2)/(1-5\eps)$. Thus, when \textproc{LazyDensityGreedy} finishes, we have $c(B_p) + c(C_p) \le \eps c(\OPT_2)$.
\end{proof}

  \begin{lemma}
  \label{claim7}
    $F(x_p) - F(x_{p - 1}) \geq \eps (1 - 12\eps) (f(\OPT) - F(x_p)) - 2\eps^2 M$.
  \end{lemma}
  \begin{proof}
    By Lemma~\ref{lem2},
    \begin{align*}
      & F(y^{(p,t+r_p)} \vee \mathbf{1}_{B_p}) - F(y^{(p, t)})\\
      &\geq \eps (1-12\eps)(F(y^{(p, t)} \vee \mathbf{1}_{\OPT_2}) - F(y^{(p, t)}) - \eps^2 M
    \end{align*}
   Additionally, by the first property in Lemma~\ref{lem1},
    \[ F(y^{(p, t)}) - F(y^{(p, 0)}) \geq \eps (F(y^{(p, 0)} \vee \mathbf{1}_{\OPT_1}) - F(y^{(p, 0)})) - \eps^2 M \]
    By combining the two inequalities, we obtain
    \begin{align*}
      & F(x_p) - F(x_{p - 1})\\
      &\geq \eps (1-12\eps)  \left(F(y^{(p, t)} \vee \mathbf{1}_{\OPT_2}) - F(y^{(p, t)}) + F(y^{(p, 0)} \vee \mathbf{1}_{\OPT_1}) - F(y^{(p, 0)}) \right) - 2\eps^2 M\\
    &\geq \eps (1-12\eps) \left(F(y^{(p, t)} \vee \mathbf{1}_{\OPT}) - F(y^{(p, t)}) \right) - 2\eps^2 M\\
    &\geq \eps (1 - 12\eps) (f(\OPT) - F(x_p)) - 2\eps^2 M
  \end{align*}
  On line 2, we used submodularity:
  \begin{align*}
    &F(y^{(p, t)} \vee \mathbf{1}_{\OPT_2}) + F(y^{(p, 0)} \vee \mathbf{1}_{\OPT_1})\\
    &\geq F((y^{(p, t)} \vee \mathbf{1}_{\OPT_2}) \vee (y^{(p, 0)} \vee \mathbf{1}_{\OPT_1})) + F((y^{(p, t)} \vee \mathbf{1}_{\OPT_2}) \wedge (y^{(p, 0)} \vee \mathbf{1}_{\OPT_1}))\\
    &= F(y^{(p, 0)} \vee \mathbf{1}_{\OPT}) + F(y^{(p, 0)})
  \end{align*}
\end{proof}

We can now complete the proof of Theorem~\ref{thm:monotone-fractional} as follows. We first show the approximation guarantee (property $(1)$). Using Lemma~\ref{claim7} and induction, we will show that, for every phase $p$, we have
  \[ f(\OPT) - F(x_p) \le (1 + \eps(1-12\eps))^{-p} f(\OPT) - 2\eps p M.\] 
In the base case $p = 0$, we have $F(x_p) = 0$ and the right-hand side is also $0$. Now consider $p \geq 1$. By rearranging the inequality in Lemma~\ref{claim7}, we obtain
  \[ (1 + \eps (1 - 12\eps)) (f(\OPT)-F(x_p)) \le f(\OPT) - F(x_{p - 1}) + 2\eps^2 M.\]
Therefore,
  \[  (f(\OPT)-F(x_p)) \le (1 + \eps (1 - 12\eps))^{-1} (f(\OPT) - F(x_{p - 1})) + 2\eps^2 M.\]
Using the inductive hypothesis, we obtain
  \[ f(\OPT) - F(x_p) \le (1 + \eps(1-12\eps))^{-p} f(\OPT) - 2\eps p M.\] 
Thus, after $1/\eps$ phases, we have
  \[ F(x_{1/\eps}) \geq \left(1 - {1 \over (1 + \eps(1-12\eps))^{1/\eps}} \right) f(\OPT) - 2 \eps M \geq \left( 1 - {1 \over e} -O(\eps)\right) f(\OPT).\]  
The second property in the theorem statement follows from the second property in Lemma~\ref{lem1}. Since the $\mathrm{LazyDensityGreedy}$ steps pick elements integrally, the fractional entries in the support of $x_{1/\eps}$ correspond to elements $\{e_{p, i} \colon p \leq 1/\eps, i \leq t\}$ that were selected on lines 6--10 of $\mathrm{KnapsackGuess}$. By the second property in Lemma~\ref{lem1}, for every phase $p$, there is a bijection $\sigma_p$ from the elements $\{e_{p, i} \colon i \leq t\}$ to $\OPT_1$ satisfying $c(e_{p, i}) \leq c(\sigma_p(e_{p, i}))$ for all $1 \leq i \leq t$. We define the mapping $\sigma: \{e_{p, i} \colon p \leq 1/\eps, i \leq t\} \times \{1, 2, \dots, 1/\eps\} \rightarrow \OPT_1$ as follows: $\sigma((e_{p, i}, p)) = \sigma_p(e_{p, i})$. The resulting mapping $\sigma$ satisfies the desired properties, since the iterations of a given phase select distinct elements and increase the value of each such element by $\eps$. 

\begin{algorithm}[t]
\begin{algorithmic}[1]
\State Let $\sigma_1,\ldots, \sigma_k$ be the fractional coordinates of $x$.
\State Sort $\sigma_1,\ldots, \sigma_k$ so that $c_{\sigma_1} \le c_{\sigma_2} \le \cdots\le c_{\sigma_k}$.
\While{$k > 0$}
\If{$k=1$}
\State $x_{\sigma_1} \gets 1$
\State \Return $x$
\EndIf
\If{$x_{\sigma_k} + x_{\sigma_{k-1}} > 1$}
\State Pick $u\in \{0, 1\}$ randomly such that $\Pr[u=1] = \frac{1-x_{\sigma_{k-1}}}{2-x_{\sigma_{k}}-x_{\sigma_{k-1}}}$
\If{$u=1$}
\State $x_{\sigma_k}\gets 1$
\State $x_{\sigma_{k-1}}\gets x_{\sigma_{k-1}}+x_{\sigma_{k}} - 1$
\State $k\gets k-1$
\Else
\State $x_{\sigma_{k-1}}\gets 1$
\State $x_{\sigma_{k}}\gets x_{\sigma_{k-1}}+x_{\sigma_{k}} - 1$
\State $\sigma_{k-1} \gets \sigma_k$
\State $k\gets k-1$
\EndIf
\Else
\State Pick $u\in \{0, 1\}$ randomly such that $\Pr[u=1] = \frac{x_{\sigma_{k}}}{x_{\sigma_{k}}+x_{\sigma_{k-1}}}$
\If{$u=1$}
\State $x_{\sigma_k}\gets x_{\sigma_{k-1}}+x_{\sigma_{k}}$
\State $x_{\sigma_{k-1}}\gets 0$
\State $\sigma_{k-1}\gets \sigma_k$
\State $k\gets k-1$
\Else
\State $x_{\sigma_{k-1}}\gets x_{\sigma_{k-1}}+x_{\sigma_{k}}$
\State $x_{\sigma_{k}}\gets 0$
\State $k\gets k-1$
\EndIf
\If{$x_{\sigma_k} = 1$}
\State $k\gets k-1$
\EndIf
\EndIf
\EndWhile
\end{algorithmic}
\caption{\textproc{Round}($x$)}
\label{alg:round}
\end{algorithm}

\clearpage

\section{Rounding algorithm and analysis of the final solution}
\label{sec:rounding}

In this section, we analyze the rounding algorithm (Algorithm~\ref{alg:round}) that rounds the fractional solution $x$ guaranteed by Theorem~\ref{thm:monotone-fractional}. We round the fractional entries of $x$ as follows. We initialize $\hat{x} = x$. For analysis purposes, we initialize $O = \OPT_1$. We sort the fractional elements in non-increasing order according to their cost. While there are fractional elements, we repeatedly move fractional mass between the two elements with highest cost as follows. Let $e_1$ and $e_2$ be the fractional elements with the highest and second-highest cost, respectively. We consider two cases:

{\bf Case 1: $\hat{x}_{e_1} + \hat{x}_{e_2} \leq 1$}. With probability $\hat{x}_{e_1} / (\hat{x}_{e_1} + \hat{x}_{e_2})$, we update $\hat{x}_{e_1} \leftarrow \hat{x}_{e_1} + \hat{x}_{e_2}$ and $\hat{x}_{e_2} \leftarrow 0$; with the remaining probability, we update $\hat{x}_{e_2} \leftarrow \hat{x}_{e_1} + \hat{x}_{e_2}$ and $\hat{x}_{e_1} \leftarrow 0$. If an element becomes integral, we remove it from the list. For analysis purposes, if an element is rounded up to $1$, we pair it up with the element $o_1 \in O$ with highest cost, and we update $O \leftarrow O \setminus \{o_1\}$. 

{\bf Case 2: $\hat{x}_{e_1} + \hat{x}_{e_2} > 1$}. With probability $(1 - \hat{x}_{e_2}) / (2 - \hat{x}_{e_1} - \hat{x}_{e_2})$, we update $\hat{x}_{e_1} \leftarrow 1$ and $\hat{x}_{e_2} \leftarrow \hat{x}_{e_1} + \hat{x}_{e_2} - 1$; with the remaining probability, we update $\hat{x}_{e_2} \leftarrow 1$ and $\hat{x}_{e_1} \leftarrow \hat{x}_{e_1} + \hat{x}_{e_2} - 1$. If an element becomes integral, we remove it from the list. For analysis purposes, if an element is rounded up to $1$, we pair it up with an element in $O$ as follows. If the element $e_1$ with the highest cost is rounded up to $1$, we pair up $e$ with the element $o_1 \in O$ with highest cost, and we update $O \leftarrow O \setminus \{o_1\}$. If the element $e_2$ with the second-highest cost is rounded up to $1$, we pair up $e_2$ with the element $o_2 \in O$ with the second-highest cost, and we update $O \leftarrow O \setminus \{o_2\}$.

If there is only one fractional entry then we can round this entry up to 1 and pair up this element with the element $o_1\in O$ with highest cost.

We now turn to the analysis of the rounding. We first show that the expected value of the rounded solution is at least $F(x)$. We then show that the cost of the fractional elements that were rounded up to $1$ is at most $c(\OPT_1)$, thus ensuring that the final rounded solution is feasible. 

\begin{lemma}
  $\Ex[F(\hat{x})] \geq F(x)$.
\end{lemma}
\begin{proof}
  Note that each iteration updates the solution as follows: $\hat{x}' = \hat{x} + \delta (\mathbf{1}_{e_1} - \mathbf{1}_{e_2})$, where $\delta$ is a random value satisfying $\Ex_{\delta}[\hat{x}'] = \hat{x}$. The multilinear extension is convex along the direction $\mathbf{1}_{e} - \mathbf{1}_{e'}$ for every pair of elements $e$ and $e'$. Therefore $\Ex_{\delta}[F(\hat{x}')] \geq F(\Ex_{\delta}[\hat{x}']) = F(\hat{x})$, and the claim follows by induction.
\end{proof}

\begin{lemma}
  Let $\hat{E}$ be the set of elements corresponding to the fractional entries that were rounded to $1$. We have $c(\hat{E}) \leq c(\OPT_1)$. 
\end{lemma}
\begin{proof}
  The lemma follows from the following invariant maintained by the algorithm for the partially rounded solution $\hat{x}$ and the set $O \subseteq \OPT_1$:

  {\bf Invariant:}
  Let $o_1, o_2, \dots, o_p$ be the elements of $O$, labeled such that $c_{o_1} \geq c_{o_2} \geq \dots \geq c_{o_p}$.
  Let $e_1, e_2, \dots, e_{\ell}$ be the elements corresponding to the fractional entries of $\hat{x}$, labeled such that $c_{e_1} \geq c_{e_2} \geq \dots \geq c_{e_{\ell}}$. We define the following grouping of the elements $e_1, e_2, \dots, e_{\ell}$ where each group contributes a fractional mass of $1$ and each element belongs to at most two groups. Consider the interval $[0, \sum_{i = 1}^{\ell} x_{e_i}]$ that is divided among the elements as follows: $[0, x_{e_1})$ corresponds to $e_1$ and, for all $2 \leq i \leq \ell$, $[\sum_{j = 1}^{i - 1} x_{e_j}, \sum_{j = 1}^{i} x_{e_j})$ corresponds to $e_i$. The elements that overlap with the interval $[i-1, i)$ define the $i$-th group.
  The invariant is that $\hat{x}$ and $O$ satisfy the following properties:
  \begin{enumerate}[$(1)$]
    \item $\sum_{i = 1}^{\ell} \hat{x}_{e_i} \leq |O|$, and
    \item for every $i \geq 1$ and each element $e$ in the $i$-th group, we have $c_e \leq c_{o_i}$.
  \end{enumerate}
  We will show the invariant using induction on the number of iterations. We start by showing the invariant at the beginning of the rounding algorithm. We can show the invariant for $x$ and $\OPT_1$ using Theorem~\ref{thm:monotone-fractional}.

  \begin{claim}
    The invariant holds for $x$ and $\OPT_1$.
  \end{claim}
  \begin{proof}
    Recall that each phase $p$ of the $\mathrm{KnapsackGuess}$ algorithm selects a set $A_p$ of elements and it increases the values of each of these elements by $\eps$. Thus the fractional value $x_{e_i}$ of each element $e_i \in E$ is equal to $\eps$ times the number of phases $p$ such that $e_i \in A_p$. Moreover, by Theorem~\ref{thm:monotone-fractional}, there is a mapping $\sigma: \{e_1, \dots, e_{\ell}\} \times \{1, 2, \dots, 1/\eps\} \rightarrow \OPT_1 $ such that, for each phase $p$ such that $e_i \in A_p$, $\sigma(e_i, p)$ exists and $c(e_i) \leq c(\sigma(e_i, p))$.

    We can think of each element $e_i$ having $x_{e, i} / \eps$ copies and each element $o \in \OPT_1$ having $|\sigma^{-1}(o)| \leq 1/\eps$ copies. By letting $\tilde{E}$ and $\tilde{O}$ be the copies of the elements in $E$ and $\OPT_1$ (respectively), we can equivalently view $\sigma$ as a bijection between $\tilde{E}$ and $\tilde{O}$ with the property that, if $\sigma((e, i)) = (o, j)$ then $c(e) \leq c(o)$. We may also assume that the elements of $O$ with the highest costs have $1/\eps$ copies, i.e., there exists an index $p'$ such that $o_1, \dots, o_{p'}$ have $1/\eps$ copies and $o_{p' + 1}, \dots, o_p$ have zero copies; we can ensure this property by reassigning pairs in $\tilde{E}$ to elements of $O$ with higher cost. Thus, if we sort $\tilde{E}$ and $\tilde{O}$ in non-increasing order according to costs, $\sigma$ maps the first $1/\eps$ elements of $\tilde{E}$ to $o_1$, the next $1/\eps$ elements to $o_2$, etc. Since the $i$-th consecutive block of $1/\eps$ elements of $\tilde{E}$ represents the fractional mass of the $i$-th group of elements, the second property of the invariant follows. The first property of the invariant follows from the fact that ${\|x\|_1 \over \eps} = |\tilde{E}| = |\tilde{O}| \leq {|\OPT_1| \over \eps}$.
  \end{proof}

  Now consider some iteration of the rounding algorithm, and suppose that the invariant holds at the beginning of the iteration. The invariant guarantees that the total fractional mass $\|\hat{x}\|_1$ is at most $|O|$ and, if we sort the fractional elements in non-increasing order according to the cost, the first unit of fractional mass can be assigned to the element $o_1$ with highest cost in $O$, the next unit of fractional mass can be assigned to the element $o_2$ with second-highest cost in $O$, etc. We will use such an assignment to argue that the invariant is preserved.

  Suppose we are in Case 1, i.e., $\hat{x}_{e_1} + \hat{x}_{e_2} \leq 1$, where $e_1$ and $e_2$ are the fractional elements with the highest and second-highest cost. Let $o_1$ be the element of $O$ with the highest cost. Since $\hat{x}_{e_1} + \hat{x}_{e_2} \leq 1$, it follows from the invariant that the entire fractional mass of $\hat{x}_{e_1} + \hat{x}_{e_2}$ is assigned to $o_1$. Since the rounding step moves fractional mass between $e_1$ and $e_2$, this property will continue to hold after the rounding step. If neither $e_1$ nor $e_2$ is rounded to $1$, the updated fractional solution clearly satisfies the invariant. Therefore we may assume that one of $e_1, e_2$ is rounded to $1$, and thus we must have had $\hat{x}_{e_1} + \hat{x}_{e_2} = 1$ before the rounding. Since $o_1$ is assigned a fractional mass of $1$ in total, $e_1$ and $e_2$ are the only elements assigned to $o_1$. Therefore, after removing $o_1$, $e_1$, and $e_2$, the remaining fractional entries and the set $O \setminus \{o_1\}$ satisfy the invariant.

  Suppose we are in Case 2, i.e., $1 < \hat{x}_{e_1} + \hat{x}_{e_2} \leq 2$, where $e_1$ and $e_2$ are the fractional elements with the highest and second-highest cost, respectively. Let $o_1$ and $o_2$ be the elements of $O$ with the highest and second-highest cost, respectively. It follows from the invariant that the fractional mass $\hat{x}_{e_1} + \hat{x}_{e_2}$ is assigned to $o_1$ and $o_2$ as follows: the $1$ unit of fractional mass assigned to $o_1$ is comprised of $\hat{x}_{e_1}$ from $e_1$ and $1 - \hat{x}_{e_2}$ from $e_2$, and $o_2$ is assigned the remaining $\hat{x}_{e_1} + \hat{x}_{e_2} - 1$ fractional mass of $e_2$. The rounding step either rounds $e_1$ to $1$ by moving $1 - \hat{x}_{e_1}$ mass from $e_2$ to $e_1$ or it rounds $e_2$ to $1$ by moving $1 - \hat{x}_{e_2}$ mass from $e_1$ to $e_2$. In the former case, after removing $e_1$ and $o_1$, the remaining fractional entries and the set $O \setminus \{o_1\}$ satisfy the invariant. Therefore we may assume that it is the latter, i.e., we round $e_2$ to $1$ and we remove $e_2$ and $o_2$. In this case, the fractional values on the elements $e_3, e_4, \dots$ move forward by $1 - \hat{x}_{e_2}$ to fill in the space vacated by $e_2$. We can also move forward their assignment to $O \setminus \{o_2\}$: $e_1$ remains entirely assigned to $o_1$ as before, and the assignment of each of the elements $e_3, e_4, \dots$ is shifted forward. Since we remove one unit from both the total fractional mass and $O$, every remaining element becomes assigned to an element of $O \setminus \{o_2\}$ whose cost is at least as much as the element of $O$ that it was previously assigned. Therefore the invariant is preserved. 
\end{proof}

\appendix

\section{Omitted proofs}
\label{app:omitted}

\begin{lemma}
  For every $o \in \OPT_2$, we have $f(\OPT_1 \cup \{o\}) - f(\OPT_1) \leq \eps^3 f(\OPT_1)$.
\end{lemma}
\begin{proof}
  Recall that $\OPT_1$ is comprised of the first $t = 1/\eps^3$ items $\{o_1, \ldots, o_t\}$ in the Greedy ordering of $\OPT$. We have
  \begin{align*}
    f(\OPT_1) &= f(\OPT_1) - f(\emptyset)\\
    &= \sum_{i = 1}^t (f(\{o_1, \ldots, o_i\}) - f(\{o_1, \ldots, o_{i - 1}))\\
    &\geq \sum_{i = 1}^t (f(\{o_1, \ldots, o_{i - 1}\} \cup \{o\}) - f(\{o_1, \ldots, o_{i- 1}\}))\\
    &\geq t \cdot (f(\OPT_1 \cup \{o\}) - f(\OPT_1)).
  \end{align*}
  The first inequality follows from the definition of $o_i$, and the second inequality follows from submodularity.
\end{proof}

\begin{corollary}
  Let $\OPT'_2$ be the subset of $\OPT_2$ consisting of all the elements $o \in \OPT_2$ of cost $c_o \leq \eps^2(1 - c(\OPT_1))$. We have $f(\OPT_1 \cup \OPT'_2) \geq (1 - \eps) f(\OPT)$.
\end{corollary}
\begin{proof}
  Since $c(\OPT_1) + c(\OPT_2) \leq 1$, there are at most $\eps^2$ items in $\OPT_2$ with cost greater than $\eps^2 (1 - c(\OPT_1))$. Since each of them has marginal value on top of $\OPT_1$ of at most $\eps^3 f(\OPT)$, the claim follows.
\end{proof}

\bibliographystyle{abbrv}
\bibliography{submodular}
\end{document}